\documentclass[a4paper,11pt,reqno]{article}
\usepackage[utf8]{inputenc}

\title{Cryptanalysis of a System based on \\ Twisted Dihedral Group Algebras}

\author{Simran Tinani\thanks{This research is supported by armasuisse Science and Technology.}}

\usepackage{graphicx}
\usepackage{amsmath, amssymb}
\usepackage[utf8]{inputenc}
\usepackage{subcaption}
\usepackage{tabularx}
\usepackage{amsmath,amssymb,amsthm}
\usepackage[english]{babel}
\usepackage{mathtools}
\usepackage{algorithmic}
\usepackage{enumerate}
\usepackage{hyperref}
\usepackage[numbers]{natbib}
\usepackage[dvipsnames]{xcolor}
\usepackage{color}
\usepackage{comment}

\usepackage{hyperref}
\hypersetup{
  colorlinks   = true, 
  urlcolor     = blue, 
  linkcolor    = OliveGreen, 
  citecolor   = red 
}

\usepackage{tabularx}
\usepackage{booktabs}
\usepackage{makecell}
\usepackage{longtable}
\usepackage[ruled,vlined]{algorithm2e}
\usepackage{algorithmic}
\usepackage{listings}
\usepackage{etoolbox}
\usepackage{lipsum}
\usepackage[left =2.7cm, right = 2.7 cm]{geometry}

\usepackage{url}

\usepackage{amssymb}
\usepackage{amsfonts}
\usepackage{amsmath}
\usepackage{comment}
\usepackage{tabularx}
\usepackage{booktabs}
\usepackage{makecell}
\usepackage{longtable}
\usepackage{graphicx}
\usepackage{caption}
\usepackage{lipsum}
\usepackage{xcolor}
\usepackage{parskip}
\usepackage{caption}
\usepackage{floatrow}
\usepackage[ruled,vlined]{algorithm2e}
\usepackage{algorithmic}
\usepackage{hyperref, enumitem}

\DeclareMathOperator{\Mat}{Mat}

\newcommand{\F}{\mathbb{F}}
\newcommand{\N}{\mathbb{N}}

\newcommand{\pk}{\mathrm{pk}}

\newfloatcommand{capbtabbox}{table}[][\FBwidth]
\newcommand{\fqst}{\mathbb{F}_{q}^*}
\newcommand{\fqa}{\mathbb{F}_{q}^\alpha}

\newcommand{\fqm}{{\mathbb{F}_{q^m}}}

\newcommand{\fqn}{\mathbb{F}_{q}^n}
\newcommand{\fq}{{\mathbb{F}_q}}

\newtheorem{protocol}{Protocol}

\newtheorem{lemma}{Lemma}
\newtheorem{proposition}{Proposition}
\newtheorem{corollary}{Corollary}
\theoremstyle{definition}

\newtheorem{definition}{Definition}
\newtheorem{example}{Example}
\theoremstyle{remark}
\newtheorem{remark}{Remark}

\begin{document}
\maketitle

\begin{abstract} Several cryptographic protocols constructed based on less-known algorithmic problems, such as those in non-commutative groups, group rings, semigroups, etc., which claim quantum security, have been broken through classical reduction methods within their specific proposed platforms. A rigorous examination of the complexity of these algorithmic problems is therefore an important topic of research. In this paper, we present a cryptanalysis of a public key exchange system based on a decomposition-type problem in the so-called twisted group algebras of the dihedral group $D_{2n}$ over a finite field $\fq$. Our method of analysis relies on an algebraic reduction of the original problem to a set of equations over $\fq$ involving circulant matrices, and a subsequent solution to these equations. Our attack runs in polynomial time and succeeds with probability at least $90$ percent for the parameter values provided by the authors. We also show that the underlying algorithmic problem, while based on a non-commutative structure, may be formulated as a commutative semigroup action problem.
\end{abstract}

\section{Introduction}
The design of efficient cryptographic systems that resist quantum attacks presently constitutes the important area of research called post-quantum cryptography. Non-commutative structures such as nonabelian groups, group rings, semigroups, etc., along with pertinent algorithmic problems, have been used for the construction of public key cryptosystems in a plethora of works in this field. Two algorithmic problems that have found great mention in this realm are the so-called conjugacy search problem and the decomposition problem (see \cite{shpilrain2005new}, \cite{cavallo}, \cite{Gu2014ConjugacySB}). Since such problems in general cannot be formulated as a version of the hidden subgroup problem in a finite abelian group, they have been suggested to render the corresponding cryptographic systems secure from known quantum attacks. However, specific instances of these problems are often solvable through other classical methods and do not have the presumed complexity in the specific suggested platform (see, for instance \cite{complex}, \cite{ben2018cryptanalysis}). Several linear algebra attacks on such cryptosystems have been devised that retrieve the shared key, often without solving the underlying algorithmic problem \cite{tsaban2015polynomial}, \cite{myasnikov2015linear}.

In \cite{orig}, the authors construct a key exchange system based on so-called twisted group algebras over a finite field $\fq$, which are similar to group algebras but have a more complicated multiplicative structure. Group algebras have found mention in some other proposed public key cryptographic schemes. In \cite{mat-gpring}, the authors construct a key exchange protocol based on the discrete logarithm problem in the semigroup $\Mat_3(\F_7[S_5])$ of $3\times 3$ matrices over the group ring $\F_7[S_5]$, where $S_5$ is the group of permutation on five symbols. In \cite{quantum-gpring}, an attack was devised by showing that $\Mat_3(\F_7[S_5])$
embeds into $\Mat_{360}(\F_7)$, for which the discrete logarithm problem can then be solved using the method in \cite{menezes} adapted to singular matrices.
The attack in \cite{efte} on the same system uses the fact that the algebra $\F_7[S_5]$ is semisimple, and so by
Maschke’s theorem it is isomorphic to a direct sum
of matrix algebras over $\F_7$. 

The authors of \cite{orig} assert that since Maschke’s Theorem is valid also for twisted group algebras, a similar attack might break the underlying problem of their system. However, to resolve this they choose $q$ such that the twisted group algebra is not semisimple. Further, they assert that the general methods of cryptanalysis in \cite{Romankov2017AGE} and \cite{Romankov-gen}, which require the construction of bases over some vector spaces, do not apply to their system. This is attributed to the facts that the twisted group algebra is not a group under the twisted multiplication and that there is an added dimension of non-commutativity with the twisted multiplication. 

The underlying platform of the system in \cite{orig} is a twisted group algebra of the dihedral group $D_{2n}$ over a finite field $\fq$ with twisted multiplication defined with the help of a function called a 2-cocycle. The 2-cocycle $\alpha$ is chosen by the authors such that $\fqa D_{2n}$ and $\fq D_{2n}$ are not isomorphic, so that one is no longer working over a group algebra. Some recent relevant works on twisted group algebras are \cite{rings} and \cite{twistedgpcodes}. In \cite{twistedgpcodes}, the authors study right ideals of twisted group algebras, endowing them with a natural distance and thus studying them as codes; they show that that all perfect linear codes are twisted group codes. In \cite{rings}, the authors use twisted dihedral group rings as a platform for a public key protocol as a non-commutative variation of the Diffie-Hellman protocol. This protocol has a similar platform to the one in \cite{orig}, but with the twisted multiplication and 2-cocycle defined differently. The authors show in \cite{orig} that the twisted group algebra platforms are structurally different. 

 The security of the protocol in \cite{orig} relies on a newly introduced algorithmic assumption, which the authors call Dihedral Product Decomposition (DPD) Assumption. Under this assumption, the authors prove that their  protocol is session-key secure in the authenticated-links adversarial model of Canetti and Krawczyk \cite{canetti}. The underlying algorithmic problem can be seen as a special form of the decomposition problem over the multiplicative monoid of an algebra $A$: given $(x, y) \in A$ and $ S \subseteq G$, the problem is to find $z_1, z_2 \in S$ such that $y = z_1 x z_2$. The Dihedral Product Decomposition Problem constitutes finding $(z_1,z_2)$ given $z_1 x z_2$ and $x$ in the platform, where $z_1$ and $z_2$ lie in specific predefined subalgebras of $\fqa D_{2n}$. It is therefore a more restricted version of the general decomposition problem in the platform. The authors claim that the protocol proposed is quantum-safe, with justification based on the fact that the decomposition problem is a generalization of the conjugacy search problem, which is believed to be difficult even for quantum computers, in certain platform groups. 


In this paper we show that in most cases, the underlying Dihedral Product Decomposition (DPD) Problem can be solved algebraically with a classical polynomial time algorithm. As a result, the Dihedral Product Decomposition Assumption does not hold, and the security of the system breaks down completely. We do this by producing an algebraic reduction of the original problem to a set of equations over $\fq$ involving circulant matrices, which we show can be solved in polynomial time in a majority of cases. We show that our algorithm succeeds with probability $1-(1-\frac{1}{q})^2$, which gives a lower bound of a 90 percent success rate with the values of $q$ and $n$ proposed by the authors. We also show that the underlying DPD problem may be formulated as a semigroup action problem \cite{maze}, with multiplication in the multiplicative monoid of a twisted dihedral group algebra. Some other protocols using this method have been proposed in \cite{gp-comm}, \cite{mat-gpring}, \cite{maze}.

The paper is structured as follows. In Section~\ref{struc} we describe the structure and some properties of the underlying platform, viz. the twisted group algebra $\fqa D_{2n}$, closely following the results of \cite{orig}. In Section~\ref{keyexch-sec}, we describe the key exchange protocol proposed in \cite{orig} and state the DPD problem, which forms the basis of its security assumption. We show that despite the use of a non-commutative structure, this algorithmic problem is equivalent to a commutative semigroup action problem. In Section~\ref{circu}, we present some background definitions and results on circulant matrices, which are needed for our reduction and cryptanalysis. In Section~\ref{rdxn}, we describe an algebraic reduction of the DPD problem to a set of simultaneous equations over $\fq$ and show that in a majority of cases, they can be solved by linear algebra in polynomial time. Using these results, we provide a polynomial time algorithm which performs the cryptanalysis of the system of \cite{orig}. 

Throughout, we let $\F$ denote a field, $G$ denote a finite group and $\fq$ denote the finite field with $q$ elements, where $q$ is a power of a prime. Also let $\fqst=\fq\setminus\{0\}$. We denote by $D_{2n} $ the dihedral group of size $2n$. 

\section{Structure of the Platform}\label{struc}

\begin{definition}[Group Algebra]\label{gpalgebra} The group algebra $\F[G]$ is the set of the formal sums $\sum\limits_{g\in G} a_g g$, with $a_g\in \F$, $g\in G$. Addition is defined componentwise: $\sum\limits_{g\in G} a_g g + \sum\limits_{g\in G} b_g g:= \sum\limits_{g\in G} (a_g+b_g) g$. Multiplication is defined as $\sum\limits_{g\in G} a_g g \cdot \sum\limits_{g\in G} b_g g:= \sum\limits_{g\in G}\sum\limits_{h\in G} (a_gb_h) gh = \sum\limits_{k\in G} \sum\limits_{g\in G, \; h\in G: gh=k} a_gb_h k$.
\end{definition}

Clearly, $\F[G]$ is an algebra over $\F$ with dimension $|G|$. If $G$ is non-commutative, so is $\F[G]$.

In \cite{sym11081019}, a new form of multiplication on the $\F$-vector space $\F[G]$ is described, which produces what are called twisted group algebras, using the concept of 2-cocycles.

\begin{definition}[2-Cocycle]
 A map $\alpha:G\times  G \rightarrow  \fqst$ is called
 a 2-cocycle
of $G$ if $\alpha(1, 1) = 1$ and for all $g, h, k \in  G$ we have $\alpha(g, hk)\alpha(h, k) =
\alpha(gh, k)\alpha(g, h)$.
\end{definition}

\begin{definition}[Twisted Group Algebra]
 Let $\alpha$ be a 2-cocycle of $G$. The twisted group algebra $\F^\alpha G$ is
the set of all formal sums $\sum\limits_{g\in G}
 a_g g$, where $a_g \in \F$, with the following twisted multiplication:
$g \cdot h =\alpha(g, h)g h$, for $g, h\in G$. The multiplication rule extends linearly to all elements of the algebra: $(\sum\limits_{g\in G}a_g g)\cdot (\sum\limits_{h\in G}b_h h) = \sum\limits_{g\in G}\sum\limits_{h\in G}a_g b_h \alpha(g, h) gh$. Addition is given componentwise as in Definition~\ref{gpalgebra}. \end{definition}

\begin{remark}
 Throughout the rest of the paper, we will be concerned with twisted group algebras, and so it is understood that the product $(\sum\limits_{g\in G}a_g g)\cdot (\sum\limits_{h\in G}a_h h)$ denotes twisted multiplication. Further, we will usually omit the $\cdot$ symbol, so that multiplication in the group $G$ and in the twisted group algebra are not differentiated by operation notation. To avoid confusion we ensure that the symbols used for elements of the group and group algebra do not intersect. \end{remark}
 

Denote the set of all 2-cocycles of $G$ into $\fq$ by $Z^2(G, \fqst)$. For $\alpha, \beta \in Z^2(G, \fqst)$, one may define the cocycle $\alpha\beta\in  Z^2(G, \fqst)$ by $\alpha\beta(g, h) = \alpha(g, h)\beta(g, h)$ for all $g, h \in G$. With this operation, $Z^2(G, \fqst)$ becomes a multiplicative abelian group.


\begin{definition}[Adjunct]
For an element $a =\sum\limits_{g\in G} a_g g \in \fqa G$ we define its adjunct as
$\hat{a} := \sum\limits_{g\in G}
a_g\alpha(g, g^{-1}))g^{-1}$
\end{definition} 

\subsection{A twisted dihedral group algebra}

For the rest of this paper, we set $G = D_{2n}$, where
$D_{2n} = \langle x, y : x
^n = y^2 = 1, yxy^{-1} = x^{-1}\rangle$ is the dihedral group of order $2n$. Further, we let $C_n = \langle x^i\rangle$ be the cyclic subgroup of $D_{2n}$ generated by $x$ and $\alpha$ be a
2-cocycle of $D_{2n}$. 

The following lemma from \cite{orig} can be verified in a straightforward manner.
\begin{lemma}[\cite{orig}]
We have 
\begin{enumerate}
    \item $\fqa D_{2n}$
is a free $\fqa C_n$-module with basis $\{1, y\}$. Therefore $\fqa D_{2n} = \fqa C_n \oplus \fqa C_n y$ as a direct sum of $\fq$-vector spaces. 
\item $\fqa C_ny \cong \fqa C_n$ as $\fqa C_n$-modules.
\item For $a \in \fqa C_n y$, $ab \in \fqa C_n$ if $b \in \fqa C_n y$ and $ab \in \fqa C_n y$ if $b \in \fqa C_n$.
\item If $a \in  \fqa C_n $, then $\Hat{a} \in \fqa C_n $. Similarly, if $a \in \fqa C_n y$, then $\Hat{a} \in \fqa C_n y$.
\end{enumerate}
\end{lemma} 

\begin{definition} \begin{enumerate}
    \item For a 2-cocycle $\alpha$ of $D_{2n}$ we define the reversible subspace
of $\fqa C_ny$ as the vector subspace
\[\Gamma_\alpha = \{a =\sum\limits_{i=0}^{n-1}a_ix^iy \in \fqa C_ny \mid a_i = a_{n-i}\text{ for } i = 1, \ldots, n-1\}.\]
\item Define a map $\psi : \fqa C_n y \rightarrow \fqa C_n $ as follows. Given $a =
\sum\limits_{i=0}^{n-1}
 a_ix^iy \in \fqa C_ny $ we define $\psi(a) = \sum\limits_{i=0}^{n-1} a_ix^i \in \fqa C_n$.
Clearly, $\psi$ is an $\fq$-linear isomorphism. \end{enumerate}
\end{definition}
In this paper, we will refer to an element $\sum\limits_{i=1}^{n-1} a_ix^i y$ of the reversible subspace $\Gamma_\alpha$ as a {reversible element} of $\fqa C_n y$ and to the corresponding vector $(a_0, \ldots, a_{n-1}) \in \fqn$ as a {reversible vector}.

\begin{lemma}[\cite{orig}]\label{commutativity} Let $\alpha$ be a 2-cocycle of $D_{2n}$. Then we have
\begin{enumerate}
    \item If \begin{equation}\label{cond1} \alpha(x^i, x^{j-i}) = \alpha(x^{j-i}, x^i)\end{equation}
for all $i, j \in  \{0,\ldots, n - 1\}$, then $ab = ba$ for $a, b \in \fqa C_n$.
\item If \begin{equation}\label{cond2} \alpha(x^{i-j}y,x^{i-j}y)\alpha (x^iy, x^{i-j}y) =\alpha(x^{n-i}y, x^{n-i} y)\alpha(x^{j-i}
y, x^{n-i}y)\end{equation} for all $i, j \in  \{0,\ldots, n- 1\}$, then $a\Hat{b} = b\Hat{a}$ for $a, b \in \Gamma_\alpha$.
\end{enumerate}
\end{lemma}

The following lemma provides an explicit construction of the 2-cocycle that will be used throughout in the cryptographic construction of \cite{orig}. 

\begin{lemma}[\cite{orig}]
 Let $\lambda \in \fqst=\fq\setminus\{0\}$. The map $\alpha_\lambda : D_{2n} \times D_{2n}\rightarrow \fqst$ defined by \begin{align}\label{cocycle}
    \alpha_\lambda(g, h) &= \lambda \ \text{for} \ g = x^iy, \ h  = x^jy \ \text{with} \ i, j \in \{0, \ldots, n -1\} \ \text{and} \nonumber \\
\alpha_\lambda(g, h) &= 1 \  \text{otherwise} 
 \end{align} is a 2-cocycle. Further, $\alpha_\lambda$ satisfies the two conditions \eqref{cond1} and \eqref{cond2}.
\end{lemma}

\begin{proof}
By definition, $\alpha_\lambda(1, 1) = 1$. Thus one only needs to verify that $\alpha_\lambda(g, h)\alpha_\lambda(gh, k) =
\alpha_\lambda(g, hk)\alpha_\lambda(h, k)$ for all $g, h, k \in D_{2n}$. Write $h = x^{j_1} y^{k_1}$ and
$k = x^{j_2} y^{k_2}$
with $i, j_1, j_2\in \{0, \ldots, n - 1\}$. The condition may then be directly verified separately in a straightforward way for the two possible cases $g=x^i$ and $g = x^iy$. The fact that $\alpha_\lambda$ satisfies conditions \eqref{cond1} and \eqref{cond2} follows from the definition.
\end{proof}

\begin{lemma}[\cite{orig}]
$\fq D_{2n}$ and $\fq^{\alpha_\lambda} D_{2n}$ are isomorphic if and only if $\lambda$ is a square in $\fq$, i.e. if and only if $\lambda^{(q -1)/2} = 1$.
\end{lemma}

\begin{lemma}[\cite{orig}]
If $\lambda_1,\lambda_2$ are not squares in $\fq$, then $\fq^{\alpha_{\lambda_1}} D_{2n}$ and $\fq^{\alpha_{\lambda_2}} D_{2n}$ are isomorphic.
\end{lemma}


From Lemma~\ref{commutativity} we thus have that for the choice $\alpha=\alpha_\lambda$ of 2-cocycle, the multiplicative ring of $\fqa C_n$ is commutative, and that $a\hat{b}=b\hat{a}$ for all $a,b \in \Gamma_\alpha$. The form \eqref{cocycle} of $\alpha=\alpha_\lambda$ is adopted throughout for the cryptosystem in \cite{orig} and thus we restrict our study to this cocycle. Thus, henceforth we take $\alpha=\alpha_\lambda$. 



\section{The key exchange protocol}\label{keyexch-sec}
Having described the relevant structural properties of the underlying platform, we now describe the key exchange protocol in \cite{orig}. This uses two-sided multiplications in $\fqa D_{2n}$.

 \subsection{Public parameters}
 
 \begin{enumerate}
     \item A number $m \in \N$ and a prime $p>2$ with $p\mid 2n$ and set $q = p^m$.
     \item A 2-cocycle $\alpha=\alpha_\lambda$ for a non-square $\lambda$
in $\fq$.  This
ensures that the platform  $\fqa D_{2n}$ is not isomorphic to $\fq D_{2n}$. 
\item An element $h = h_1 + h_2$ for a random $0\neq h_1 \in \fqa C_n$ and a random $0\neq h_2 \in \fqa C_ny$. (Clearly, since $h$ is public, so are $h_1$ and $h_2$.)
 \end{enumerate}

Protocol~\ref{key-exch} describes the key exchange protocol of \cite{orig}.

\begin{protocol}\label{key-exch} 
\begin{enumerate}
    \item  Alice chooses a secret pair $(s_1, t_1)\in \fqa C_n\times \Gamma_{\alpha}$, and sends $ \pk_A = s_1ht_1$ to Bob.
 \item  Bob chooses a secret pair $(s_2, t_2 )\in \fqa C_n\times \Gamma_{\alpha}$ and sends $\pk_B = s_2ht_2 $ to Alice.
 \item Alice computes $K_A = s_1 pk_B \hat{t_1} $,
 \item Bob computes $K_B = s_2 pk_A\hat{t_2 }$
 \item The shared key is $K=K_A=K_B$
\end{enumerate}
\end{protocol}
 The authors' proposed values for parameters $q$ and $n$ are $q=n=19$, $q=n=23$, $q=n=31$, $q=n=41$.

\subsection{Correctness}
It is easy to show that within an uncorrupted session, both Alice and Bob establish the same key.
Indeed, because of the choice of $\alpha=\alpha_\lambda$, we have $s_is_j = s_js_i$ in $\fqa C_n$ and $t_i \hat{t_j} =t_j \hat{t_i}$ in $\fqa C_ny $ for $i, j \in \{1,2\}$, so
\[K_A = s_1 \pk_B
\hat{t_1} = s_1s_2 h t_2\hat{t_1} = s_2s_1 ht_1\hat{t_2}= s_2\pk_A\hat{t_2}
 = K_B.\]
 
 \subsection{Security Assumption}
 The security of the protocol depends on the assumption of the difficulty of the following algorithmic problem. 
 
 \begin{definition}[Dihedral Product Decomposition (DPD) Problem]
 
Let $(s, t)\in \fqa C_n \times \Gamma_{\alpha_\lambda}$ be a secret key.  Given a public element $h=h_1+h_2 \in  \fqa D_{2n}, \ h_1\in \fqa C_n, \ h_2\in \fqa C_n y $, and a public key $\pk = s ht $, the DPD problem requires an adversary to compute $(\Tilde{s}, \Tilde{t})\in  \fqa C_n \times \Gamma_{\alpha}$ such that $\pk = \Tilde{s} h\Tilde{t}$. 
 
 \end{definition}

Let $(\tilde{s}, \tilde{t})$ be the output of an adversary $\mathcal{A}$ attempting to solve the DPD problem for $\fqa D_{2n}$. The authors define $\mathcal{A}$’s advantage $DPD_{adv}[\mathcal{A}, \ \fqa D_{2n}]$ in solving the DPD problem  as the probability that $\tilde{s}h\tilde{t} = sht$. 
 
\begin{definition}[DPD Assumption]
 The
DPD assumption is said to hold for $\fqa D_{2n}$ if for all
efficient adversaries $\mathcal{A}$ the quantity $DPD_{adv}[\mathcal{A}, \ \fqa D_{2n}]$ is negligible.
\end{definition}

In Section~\ref{rdxn}, we provide a cryptanalysis of Protocol~\ref{key-exch} by solving the DPD problem. We show that in most cases, a polynomial time solution is possible, and so the DPD assumption does not hold. For our method of cryptanalysis, we need some prerequisites on circulant matrices, which we provide in the next section. However, we first show below how the DPD problem can be formulated as a special case of a commutative semigroup action problem, in the framework introduced in \cite{maze}.
 
\subsubsection{DPD problem as a commutative semigroup action}

The authors of \cite{orig} assert that given a fixed $h \in  \fqa D_{2n}$, the set of keys $\{sht \mid (s, t)\in \fqa C_n \times \Gamma_{\alpha}\}$  is not even a semigroup under the twisted algebra multiplication. From this observation, they claim that their system is immune to the quantum cycle-finding algorithm of Shor \cite{shor} which is known to solve the hidden subgroup problem in abelian groups.  

Further, the security of the system of \cite{orig} is based on the presence of a non-commutative multiplication in the twisted group algebra.  However, we now show that the DPD problem can be formulated as a commutative semigroup action problem, and so any classical or quantum solution to the latter also applies to the former. In \cite{monico}, a Pollard-rho type square root algorithm was provided to solve an abelian group action problem, whereas the possibility for a modification to the commutative semigroup case was left open.

As observed before, the cocycle $\alpha=\alpha_\lambda$ satisfies conditions \eqref{cond1} and \eqref{cond2}. Thus, $ab = ba$ for $a, b \in \fqa C_n$ and $a\Hat{b} = b\Hat{a}$ for $a, b \in \Gamma_\alpha$. In particular, $\fqa C_n$ is a commutative subalgebra of $\fqa D_{2n}$. 
Recall the $\fq$-linear isomorphism $\psi: \fqa C_n \rightarrow \fqa C_ny$ given by $\psi(a) = \sum\limits_{i=0}^{n-1} a_ix^i \in \fqa C_n$ for $a =
\sum\limits_{i=0}^{n-1}
 a_ix^iy \in \fqa C_ny$. 
 Notice that $\psi(\Gamma_{\alpha})$ is a commutative semigroup under the multiplication defined by $\psi(t)\star \psi(t'):=t \hat{t'}\in \psi(\Gamma_{\alpha})$.

We can now look at the key exchange in Protocol~\ref{key-exch} as an instance of a semigroup action problem, introduced in \cite{maze}.

\begin{definition}[Semigroup Action Problem] Let $S$ be any semigroup acting on a set $X$ \begin{align*}
    S \times X\rightarrow X \\
    (s,x)\mapsto s\cdot x
\end{align*}
Given an element $y=s\cdot x \in X$, where $x\in X$ is known and $s \in S$ is a secret, the semigroup action problem is to find some $\Tilde{s}\in S$ such that $\Tilde{s}\cdot x = y$. \end{definition}

\begin{proposition}
The commutative semigroup $\fqa C_n\times \psi(\Gamma_{\alpha})$ acts on $\fqa D_{2n}$ as follows \begin{align}\label{comm-axn}
    (\fqa C_n\times \psi(\Gamma_{\alpha})) \times \fqa D_{2n} & \rightarrow \fqa D_{2n}\nonumber \\
    (s,\psi(t))\cdot h & = sht
\end{align}
\end{proposition}
\begin{proof}
Clearly, $(1,1)\cdot h = h$ for all $h\in \fqa D_{2n}$. Further, 
\[(s,\psi(t))((s',\psi(t'))\cdot h )= ss'h t'\hat{t}=ss' h t\hat{t'}=(ss',t\hat{t'})\cdot h =(ss',\psi(t)\star\psi(t'))\cdot h.\]\end{proof}

\begin{lemma}
 The DPD problem is equivalent to the semigroup action problem for the commutative semigroup action~\eqref{comm-axn}
\end{lemma}
\begin{proof}
Clearly, $t$ and $\psi(t)$ can easily be read from each other without any significant computational cost. Suppose that given public element $h$ and public key $\pk$, the adversary can find $s,t$ such that $sht=\pk$. Then, $(s,\psi(t))$ is a solution to the SAP~\eqref{comm-axn}. Conversely, any solution $(s, \psi(t))$ of the SAP~\eqref{comm-axn} gives the solution $(s,t)$ of the DPD problem. 
\end{proof}

 The next section highlights some prerequisites on circulant matrices which will be used in the cryptanalysis of the system in Section~\ref{rdxn}.
 
\section{Circulant Matrices}\label{circu}

\begin{definition}\label{circ-def}
A matrix over $\fq$ of the form $\begin{pmatrix}c_0 & c_{n-1} & \ldots & c_1\\
  c_1 & c_0 & \ldots & c_2\\ \vdots & \vdots & \ddots & \vdots \\
  c_{n-1} & c_{n-2} & \ldots & c_0\\
  \end{pmatrix}$ with $c_i\in \fq$, is called circulant. Given a vector $\mathbf{c}=(c_0, c_1, \ldots, c_{n-1})^T\in \fqn$, we use the notation $M_\mathbf{c}$ to denote the circulant matrix $M_{\mathbf{c}}:=\begin{pmatrix}c_0 & c_{n-1} & \ldots & c_1\\
  c_1 & c_0 & \ldots & c_2\\ \vdots & \vdots & \ddots & \vdots \\
  c_{n-1} & c_{n-2} & \ldots & c_0\\
  \end{pmatrix}$.
\end{definition}

\begin{definition}
Given vectors $\mathbf{b}=(b_0, b_1, \ldots, b_{n-1})^T\in \fqn$, $\mathbf{c}=(c_0,c_1, \ldots, c_{n-1})^T\in \fqn$, define, for $0\leq \ell \leq n-1$ the constants \[z_\ell({\mathbf{b}, \mathbf{c}}) =\sum\limits_{i+j=\ell \mod n}b_ic_j =\begin{pmatrix}c_\ell, & c_{\ell-1}, & \ldots, &c_{\ell+1}\end{pmatrix}\cdot \begin{pmatrix}
  b_0 \\ b_1\\ \vdots \\ b_{n-1}   \end{pmatrix}, 0\leq\ell \leq n-1. \]Also define the vector $\mathbf{z}_{\mathbf{b},\mathbf{c}}=( z_0({\mathbf{b}, \mathbf{c}}), \ldots, z_\ell({\mathbf{b}, \mathbf{c}}), \ldots, z_{n-1}({\mathbf{b}, \mathbf{c}}))^T$. In other words, \[\mathbf{z}_{\mathbf{b},\mathbf{c}}=\begin{pmatrix}c_0 & \ldots & c_1\\
  c_1 & \ldots & c_2\\ \vdots & \ddots & \vdots \\
  c_{n-1} & \ldots & c_0\\
  \end{pmatrix}\cdot \begin{pmatrix}
  b_0 \\ b_1\\ \vdots \\ b_{n-1}   \end{pmatrix}=M_{\mathbf{c}}\cdot \mathbf{b}.\]

As in Definition~\ref{circ-def}, denote by $M_{\mathbf{z}}({\mathbf{b},\mathbf{c}})$ the  circulant matrix $M_{\mathbf{z}}({\mathbf{b},\mathbf{c}})=\begin{pmatrix}z_0({\mathbf{b}, \mathbf{c}}) & \ldots & z_1({\mathbf{b}, \mathbf{c}})\\
  z_1({\mathbf{b}, \mathbf{c}}) & \ldots & z_2({\mathbf{b}, \mathbf{c}})\\ \vdots & \ddots & \vdots \\
  z_{n-1}({\mathbf{b}, \mathbf{c}}) & \ldots & z_0({\mathbf{b}, \mathbf{c}})\\
  \end{pmatrix}$. The following result is easy to verify by direct computation.
  
  \begin{lemma}$M_{\mathbf{z}}({\mathbf{b},\mathbf{c}})=M_{\mathbf{c}} \cdot M_{\mathbf{b}} $.\end{lemma}
\end{definition}

\subsection{Probability of a circulant matrix being invertible}

We will require the invertibility of some random circulant matrices over $\fq$ for our reduction of the system. For this reason, we discuss the criteria for a random circulant matrix being invertible, and study this probability. We have the following result from \cite{prop1}.

\begin{proposition}[\cite{prop1}]
 Let $x^n-1=f_1^{\alpha_1}(x) \ldots f_\tau^{\alpha_\tau}(x)$
 be the factorization of $x^n - 1$
over $\fqm$ into powers of irreducible factors. The number of invertible circulant
matrices in $Mat_n(\fqm)$ is equal to $\prod\limits_{i=1}^{\tau}
(q^{md_i\alpha_i} - q^{md_i(\alpha_i-1)})$, where $d_i$ is the degree
of $f_i(x)$ in the factorization of $x^n- 1$.
\end{proposition}

Note that the number of circulant matrices over $\fqm$ is $q^{mn}$. As a direct consequence, the probability of a randomly chosen circulant matrix over $\fqm$ being invertible is \[\prod\limits_{i=1}^\tau \dfrac{q^{md_i\alpha_i-q^{md_i(\alpha_i-1)}}}{q^{nm}}=\prod\limits_{i=1}^\tau \left(1-\frac{1}{q^{md_i}}\right) \] 
It is now easy to see that a lower bound for this quantity is $(1-\frac{1}{q^m})^n$, which is achieved if $x^n-1$ splits into distinct linear factors, i.e. $\tau=n$, $d_i=1$, $\alpha_i=1$. Similarly, an upper bound is achieved when there is a single factor in the factorization, i.e. $\tau=1$ and $\alpha_1=n$, in which case the quantity is $(1-\frac{1}{q^m})$. Note that this upper bound is achieved when the characteristic $p$ of $\fqm$ divides $n$ ($x^n-1=(x-1)^n \mod p$). Thus, we have the following corollary. 

\begin{corollary}\label{inver}
If $p\mid n$ then the probability that a randomly chosen $n\times n$ circulant matrix over $\fq$ is invertible is $1-\frac{1}{q}$.
\end{corollary}
In \cite{orig}, the authors deliberately choose the case $p \mid n$, so as to avoid having $\fqm D_{2n}$ semisimple, and so, the probability $1-\frac{1}{q}$ applies for a random circulant matrix being invertible.

\section{Cryptanalysis}\label{rdxn}

Note that the adversary is given an equation of the form $s h t= \gamma$ over $\fqa D_{2n}$, where \begin{equation}\label{alpha,t} s=\sum\limits_{i=0}^{n-1}{a_ix^i}\in \fqa C_{n}, \  t=\sum\limits_{i=0}^{n-1}{b_ix^iy}\in \Gamma_\alpha\subseteq \fq^{\alpha_\lambda}D_{2n}\end{equation} are unknown, and $h=\sum\limits_{i=0}^{n-1}c_ix^i + \sum\limits_{i=0}^{n-1}d_ix^iy$ is known. Since $t\in \Gamma_\alpha$, the coefficients in $t$ satisfy $b_k=b_{n-k}$ for $k=1,\ldots, n-1$. We write \[\gamma=\sum\limits_{i=0}^{n-1}{v_ix^i}+\sum\limits_{i=0}^{n-1}{w_ix^iy}\] for known constants $v_i, w_i$. Substituting the above expansions into the equation $s ht=\gamma$, we have 
\begin{align*}
    (\sum\limits_{i=0}^{n-1}{a_ix^i})(\sum\limits_{i=0}^{n-1}{c_ix^i}+\sum\limits_{i=0}^{n-1}{d_ix^iy})(\sum\limits_{i=0}^{n-1}{b_ix^iy}) = \sum\limits_{i=0}^{n-1}{v_ix^i}+\sum\limits_{i=0}^{n-1}{w_ix^iy} \\
    \implies (\sum\limits_{i,j=0}^{n-1}{a_ic_jx^{i+j}}+\sum\limits_{i,j=0}^{n-1}{a_id_jx^{i+j}y})(\sum\limits_{k=0}^{n-1}{b_kx^{k}y}) = \sum\limits_{i=0}^{n-1}{v_ix^i}+\sum\limits_{i=0}^{n-1}{w_ix^iy}\\
    \implies \sum\limits_{i,j,k=0}^{n-1}{a_ic_jb_kx^{i+j+k}y} + \sum\limits_{i,j,k=0}^{n-1}{a_id_jb_k\lambda x^{i+j+k}} =\sum\limits_{i=0}^{n-1}{v_ix^i}+\sum\limits_{i=0}^{n-1}{w_ix^iy} 
\end{align*}

Comparing coefficients, we have the following two equations \begin{equation}\label{acw}
    \sum\limits_{i,j,k=0}^{n-1}{a_ic_jb_kx^{i+j+k}y} =  \sum\limits_{i=0}^{n-1}{w_ix^iy},
\end{equation} \begin{equation}\label{adv}
   \lambda\sum\limits_{i,j,k=0}^{n-1}{a_id_jb_k x^{i+j+k}} =\sum\limits_{i=0}^{n-1}{v_ix^i} 
\end{equation} 

Define vectors $\mathbf{a} = (a_0, \ldots, a_{n-1})^T$, $\mathbf{b} = (b_0, \ldots, b_{n-1})^T$, $\mathbf{c} = (c_0, \ldots, c_{n-1})^T$, $\mathbf{d} = (d_0, \ldots, d_{n-1})^T$, $\mathbf{w} = (w_0, \ldots, w_{n-1})^T$, $\mathbf{v} = (v_0, \ldots, v_{n-1})^T$ in $\fqn$. The vectors $\mathbf{a}$ and $\mathbf{b}$ are unknown to the adversary, while $\mathbf{c}$, $\mathbf{d}$, $\mathbf{v}$, and $\mathbf{w}$ are publicly known. 

\subsection{Reduction to matrix equations}


The following lemma shows that Equation~\eqref{acw} can be reduced to a matrix equation over $\fq$.

\begin{lemma}\label{lem1}
 Equation~\eqref{acw} is  equivalent to the matrix equation $M_{\mathbf{z}}({\mathbf{b},\mathbf{c}})\cdot \mathbf{a} = \mathbf{w}$ over $\fq$. 
\end{lemma}
\begin{proof}
Equating the coefficients of the basis vectors $x^iy$ in Equation~\eqref{acw}, we have 
\begin{align*}w_i = &\sum\limits_{\ell=0}^{n-1} \ \sum\limits_{(j,k) \mid j+k=\ell \mod n}{c_j}b_k a_{i-\ell} \\= & \sum\limits_{\ell=0} \ \sum\limits_{(j,k)\mid j+k=i-\ell \mod n}{c_j}b_k a_{\ell} \\= &  
\begin{pmatrix}z_i(\mathbf{b},\mathbf{c}) & z_{i-1}(\mathbf{b},\mathbf{c}) & \ldots & z_0(\mathbf{b},\mathbf{c}) & z_{n-1}(\mathbf{b},\mathbf{c}) & \ldots & z_{i+1}(\mathbf{b},\mathbf{c}) \end{pmatrix}\cdot \mathbf{a}
\end{align*}

Thus, we can rewrite Equation~\eqref{acw} 
equivalently as the system \begin{align*}
    w_0 = & \begin{pmatrix}z_0(\mathbf{b},\mathbf{c}) & z_{n-1}(\mathbf{b},\mathbf{c}) & \ldots & z_1({\mathbf{b},\mathbf{c}}) \end{pmatrix}\cdot \mathbf{a} \\
     w_1 = & \begin{pmatrix}z_1(\mathbf{b},\mathbf{c})  &z_0(\mathbf{b},\mathbf{c})  & \ldots & z_2(\mathbf{b},\mathbf{c}) \end{pmatrix}\cdot  \mathbf{a} \\
  \vdots \\
       w_{n-1}= & \begin{pmatrix}z_{n-1}({\mathbf{b},\mathbf{c}}) &z_{n-2}({\mathbf{b},\mathbf{c}}) & \ldots &z_0({\mathbf{b},\mathbf{c}})\end{pmatrix}\cdot  \mathbf{a}
\end{align*}
In other words, $M_{\mathbf{z}}({\mathbf{b},\mathbf{c}})\cdot \mathbf{a} = \mathbf{w}$.
\end{proof}

 One may similarly rewrite Equation~\eqref{adv} as above, so that we have the following lemma. 
\begin{lemma}\label{lem2}
 Equation~\eqref{adv} is  equivalent to the matrix equation $\lambda M_\mathbf{z}(\mathbf{b},\mathbf{d})\cdot \mathbf{a} = \mathbf{v}$ over $\fq$. 
\end{lemma}




Combining the results of Lemmas~\ref{lem1} and \ref{lem2},  if the vectors $\mathbf{b}$,$\mathbf{c}$ and $\mathbf{d}$ are given, then $\mathbf{a}$ is a simultaneous solution to the matrix equations $M_{\mathbf{z}}({\mathbf{b},\mathbf{c}})\cdot \mathbf{a} = \mathbf{w}$ and $\lambda M_{\mathbf{z}}({\mathbf{b},\mathbf{d}})\cdot \mathbf{a} = \mathbf{v}$.  However, a priori the vector $\mathbf{b}$ is unknown to the adversary. If we can find $\mathbf{b}$ such that this system of equations has a simultaneous solution, then we are done with reducing the DPD problem to a solving a single system of linear equations, which can be done in polynomial time. Summarizing this discussion, we have the following result. 

\begin{proposition}
 Suppose that a vector $\mathbf{b}=(b_0, \ldots, b_{n-1})$ is such that the system of simultaneous equations $\lambda M_\mathbf{z}(\mathbf{b},\mathbf{d})\mathbf{a}=\mathbf{v}$ and $M_{\mathbf{z}}({\mathbf{b},\mathbf{c}})\mathbf{a}=\mathbf{w}$ has a simultaneous solution $\mathbf{a}=(a_0, \ldots, a_{n-1})$. Then, $s=\sum\limits_{i=0}^{n-1} a_ix^i$, $t=\sum\limits_{i=0}^{n-1} b_ix^iy$ is a solution of the equation $s ht=\gamma$.
\end{proposition}

Now, for an adversary, the vectors $\mathbf{a}$ and $\mathbf{b}$ are both unknown.  We will show below that in most cases, it suffices for the adversary to fix a suitable value for $\mathbf{b}$ and then proceed to solve any one of the linear equations in Lemmas~\ref{lem1} and \ref{lem2} for $\mathbf{a}$. More precisely, we show that if $M_{\mathbf{c}}$ and $M_{\mathbf{d}}$ are invertible, then a solution is possible for any randomly chosen $\mathbf{b}\in \Gamma_\alpha$ for which the correponding circulant matrix $M_{\mathbf{b}}$ is invertible. Since the values arise from a legitimate public key, we know that there exists a vector $\mathbf{b}\in \Gamma_\alpha$ such that the equations $\lambda M_\mathbf{z}(\mathbf{b},\mathbf{d})\mathbf{a}=\mathbf{v}$ and $M_\mathbf{z}({\mathbf{b},\mathbf{c}})\mathbf{a}=\mathbf{w}$ have a simultaneous solution $\mathbf{a}$.

\begin{proposition}
 Let the vectors $\mathbf{c}$ and $\mathbf{d}$ be such that $M_{\mathbf{c}}$ and $M_{\mathbf{d}}$ are invertible. Assume that at least one simultaneous solution $(\mathbf{a},\mathbf{b})$ exists to the matrix equations $\lambda M_\mathbf{z}({\mathbf{b},\mathbf{d}}) \mathbf{a}=\mathbf{v}$ and $M_\mathbf{z}(\mathbf{b},\mathbf{c})\mathbf{a}=\mathbf{w}$. Then, for any randomly chosen $\mathbf{b}\in \Gamma_\alpha$ such that $M_{\mathbf{b}}$ is invertible, the equations $\lambda M_\mathbf{z}({\mathbf{b},\mathbf{d}}) \mathbf{a}=\mathbf{v}$ and $M_\mathbf{z}(\mathbf{b},\mathbf{c})\mathbf{a}=\mathbf{w}$ have a simultaneous solution  $\mathbf{a}$ computable in polynomial time.
\end{proposition}
\begin{proof}
Here, $\mathbf{b}$, $\mathbf{c}$, and $\mathbf{d}$ are invertible, and thus so are $M_\mathbf{z}({\mathbf{b},\mathbf{d}})=M_{\mathbf{d}}\cdot M_{\mathbf{b}}$ and $M_\mathbf{z}(\mathbf{b},\mathbf{c})=M_{\mathbf{c}}\cdot M_{\mathbf{b}}$.  Now, we know that a solution $(\mathbf{a},\mathbf{b})$ exists, and so for some vectors $\mathbf{a}$ and $\mathbf{b}$ we have
\begin{align*}
    \lambda M_{\mathbf{d}}M_{\mathbf{b}}\mathbf{a}=\mathbf{v}, \quad M_{\mathbf{c}}M_{\mathbf{b}}\mathbf{a}=\mathbf{w}, \ i.e. \;
    \lambda^{-1}M_{\mathbf{d}}^{-1}\mathbf{v} = M_{\mathbf{b}}\mathbf{a}, \ M_{\mathbf{c}}^{-1}\mathbf{w}=M_{\mathbf{b}} \mathbf{a}
\end{align*}

 So, independently of $\mathbf{a}$ and $\mathbf{b}$ we necessarily have
 \begin{align}\label{cdvw} \lambda^{-1}M_{\mathbf{d}}^{-1}\mathbf{v} = M_{\mathbf{c}}^{-1}\mathbf{w}
 \end{align} 
 Now let $\mathbf{b}$ be any random vector such that $M_{\mathbf{b}}$ is invertible. Multiplying equation \eqref{cdvw} by $M_{\mathbf{b}}^{-1}$, we get 
 \begin{align*} \lambda^{-1}M_{\mathbf{b}}^{-1}M_{\mathbf{d}}^{-1} \mathbf{v}= M_{\mathbf{b}}^{-1}M_{\mathbf{c}}^{-1}\mathbf{w} \\
     \implies \lambda^{-1}M_\mathbf{z}(\mathbf{b},\mathbf{d})^{-1}\mathbf{v} = M_\mathbf{z}({\mathbf{b},\mathbf{c}})^{-1}\mathbf{w} \\
 \end{align*}
 Setting $\mathbf{a}:=\lambda^{-1}M_\mathbf{z}(\mathbf{b},\mathbf{d})^{-1}M_\mathbf{v} = M_\mathbf{z}({\mathbf{b},\mathbf{c}})^{-1}\mathbf{w}$, we get $\mathbf{a}$ as the simultaneous solution $\lambda M_\mathbf{z}({\mathbf{b},\mathbf{d}})\mathbf{a} = \mathbf{v}$ and $M_\mathbf{z}({\mathbf{b},\mathbf{c}})\mathbf{a}=\mathbf{w}$.
\end{proof}

\subsection{The algorithm for cryptanalysis}

We have the following result.
\begin{corollary}\label{fnlr}
If $M_\mathbf{c}$ and $M_\mathbf{d}$ are invertible and $\gamma$ is a legitimate public key, then the equation $sht=\gamma$ in the unknowns $s \in \fqa C_n, \ t\in \Gamma_\alpha$ can be solved in polynomial time for a legitimate secret key $(s,t)$.  
\end{corollary}
\begin{proof}
Since $\gamma$ is a legitimate public key, a least one simultaneous solution $(\mathbf{a},\mathbf{b})$ exists (the one corresponding to the initial secret key) to the matrix equations $\lambda M_\mathbf{z}({\mathbf{b},\mathbf{d}}) \mathbf{a}=\mathbf{v}$ and $M_\mathbf{z}(\mathbf{b},\mathbf{c})\mathbf{a}=\mathbf{w}$. Now, from Corollary~\ref{inver}, a vector $\mathbf{b} \in \fqn$ such that $\mathbf{b}$ is invertible can be found in an expected $\frac{1}{1-\frac{1}{q}}$ number of steps. For the solution of the DPD problem, one further requires that the vector $\mathbf{b}$ satisfies $b_i=b_{n-1}$ for $1\leq i \leq n-1$, i.e. that $\mathbf{b}\in \Gamma_\alpha$. However, it is prudent to assume that the probability of invertibility remains approximately the same on these reversible vectors. Thus, by Proposition 4, we can set $b$ to be any vector in $\Gamma_\alpha$ such that $M_\mathbf{b}$ is invertible. The expected number of steps before such a $\mathbf{b}$ is found is $\frac{1}{1-\frac{1}{q}}$, which is very close to 1, and thus takes time $\mathcal{O}(1)$. This is also confirmed by experimental results, where randomly chosen symmetric vectors $\mathbf{b}\in \Gamma_\alpha$ were invertible in almost all trials. Once such a vector $\mathbf{b}$ is found, one computes $\mathbf{a}=\lambda^{-1}M_\mathbf{z}(\mathbf{b},\mathbf{d})^{-1}M_\mathbf{v} = M_\mathbf{z}({\mathbf{b},\mathbf{c}})^{-1}\mathbf{w}$ in polynomial time. By Proposition 3, this gives a solution to the DPD problem $sht=\gamma$.
\end{proof}

We now state an algorithm to cryptanalyze the key exchange. Its correctness follows from the above discussion.

\begin{algorithm}[H]\label{cryptanalysis}
\hspace*{\algorithmicindent} \textbf{Input} {Parameter $\lambda$ and the cocycle $\alpha=\alpha_\lambda$, public element $h=\sum\limits_{i=0}^{n-1}c_ix^i + \sum\limits_{i=0}^{n-1}d_ix^iy$, public key $\gamma=\sum\limits_{i=0}^{n-1}{v_ix^i}+\sum\limits_{i=0}^{n-1}{w_ix^iy}$.} \\
 \hspace*{\algorithmicindent} \textbf{Output} {A solution $(s,t) \in \fqa C_n \times \Gamma_\alpha$ satisfying $sht=\gamma$. This tuple is a solution to the DPD problem.}
\caption{Cryptanalysis of Key Exchange over $\fqa D_{2n}$}
\begin{algorithmic}[1]
\STATE Define vectors in $\fqn$: $\mathbf{c}=(c_0,\ldots, c_{n-1})$, $\mathbf{d}=(d_0,\ldots, d_{n-1})$, $\mathbf{v}=(v_0,\ldots, v_{n-1})$, $\mathbf{w}=(w_0,\ldots, w_{n-1})$. 
\STATE If $M_\mathbf{c}$ or $M_\mathbf{d}$ is not invertible \\
\qquad Return Fail
\STATE Pick  a vector $\mathbf{b}=(b_0,\ldots, b_{n-1})\leftarrow \Gamma_\alpha$ at random. 
\STATE If $M_\mathbf{b}$ is not invertible, repeat step (3). If it is invertible, go to step (5).
\STATE Compute $\mathbf{a}=\lambda^{-1} M_\mathbf{z}(\mathbf{b}, \mathbf{c})^{-1}\mathbf{w}=M_{\mathbf{b}}^{-1}M_{\mathbf{d}}^{-1}\mathbf{v}$.
 \STATE With $\mathbf{a}= (a_0,\ldots, a_{n-1})$, set $s=\sum_{i=0}^{n-1}a_ix^i$ and $t=\sum_{i=0}^{n-1}b_ix^i y$. 
 \STATE Return $(s,t)$.
\end{algorithmic}
\end{algorithm}

\begin{remark}
The solution $(s,t)$ to the DPD returned by Algorithm~\ref{cryptanalysis} and referenced in Corollary~\ref{fnlr} is a legitimate secret key, but not necessarily the same as the originally chosen secret key. In fact, as is clear from the discussion above, $t=\sum\limits_{i=0}^{n-1} b_ix^iy\in \Gamma_\alpha$ can be selected at random, and a solution for $s\in \fqa C_n$ is found long as $M_{\mathbf{b}}$ is invertible..
\end{remark}


Now, since $\mathbf{c}$ and $\mathbf{d}$ are random in $\fqn$, the circulant matrices $M_{\mathbf{c}}$ and $M_{\mathbf{d}}$ are invertible with high probability. The probability that the algorithm fails is the probability that at least one of them is not invertible, which is given by $1-(1-\frac{1}{q})^2$. Clearly this quantity shrinks with increasing values of $q$ and $n$. In \cite{orig} the smallest values of these parameters are $q=n=19$, for which this probability is $\approx 0.1$. Thus, Algorithm~\ref{cryptanalysis} succeeds in cryptanalyzing the system with a probability of at least 90 percent.

 An immediate corollary of the above argument is that the two-sided multiplication action \[(\fqa C_n \times \Gamma_\alpha) \times \fqa D_{2n} \rightarrow \fqa D_{2n}\]
\[(s,t)\cdot h \mapsto sht, \ s\in \fqa C_n, \ t \in \Gamma_\alpha\] is far from being injective, contrary to the assumption of the authors. In fact, for most values of $t$ and $\gamma \in \fqa D_{2n}$, there is a unique pre-image $s\in \fqa C_n$ such that $sh t=\gamma$. Thus, the probability that random choosing yields the right solution is not $1/|\fqa C_n \times \Gamma_\alpha | $, as claimed by the authors. The real probability is greater than or equal to probability that the matrices $M_\mathbf{c}$ and $M_\mathbf{d}$ are invertible and that the correct value of $s$ corresponding to $t$ is chosen, which is $\approx 1/|\fqa C_n |$ (we already saw that the probability of the matrices being invertible is very close to 1). 
From this, one also sees that the run time of an exhaustive search would be linear in $|\fqa  C_n  | =p^{nm}$, rather than in $|\fqa  C_n \times \Gamma_\alpha | =p^{nm} p^{m\lfloor \frac{n+1}{2}\rfloor }$, as claimed by the authors of \cite{orig}.

\subsection{Examples}
In this subsection, we present some examples generated by computer search, using the algebra software package SageMath \cite{sagemath}. For the structure of the twisted group algebra and the generation of the keys, we made use of the original source code of the authors. Our entire working code including the cryptanalysis can be found at: \url{https://github.com/simran-tinani/Cryptanalysis-of-twisted-group-algebra-system}

In the following examples, an element $\sum\limits_{i=0}^{n-1}a_ix^i+ \sum\limits_{i=0}^{n-1}b_ix^iy$ of $\fqa D_{2n}$ is denoted by the $2n$-tuple $(a_0, \ldots, a_{n-1}, b_0, \ldots, b_{n-1})$.

\begin{example}
For parameters $n=23,\
q=23$, $\lambda=11$, and using the notations above, consider the randomly generated public element $h$, and secret key $(s,t)\in \fqa C_n\times \Gamma_\alpha$.
\scriptsize \begin{align*}
   h= & (19,9,4,14,6,13,21,18,18,10,9,2,5,15,13,22,18,13,16,20,11,2,11,6,18,7,17,8,20,20,17,7,15,1,11,9,17,4,11,16,\\  & 5,17,19,18,19,20), \\
   s=& (20,17,20,22,18,18,11,12,2,3,18,11,2,18,3,14,10,2,13,14,3,9,17,0,0,0,0,0,0,0,0,0,0,0,0,0,0,0,0,0,0,0,0,0,0,0) \\ 
   t = & (0,0,0,0,0,0,0,0,0,0,0,0,0,0,0,0,0,0,0,0,0,0,0,22,0,14,3,2,19,4,15,1,21,3,6,6,3,21,1,15,4,19,2,3,14,0)
\end{align*}\normalsize
Using the method in Section~\ref{rdxn}, the program computed the solution $(\tilde{s}, \tilde{t})$ to the DPD, where
\scriptsize
\begin{align*}
   \tilde{s} = & (13,16,5,1,21,1,2,8,17,2,12,11,4,0,20,7,19,16,3,14,22,6,2,0,0,0,0,0,0,0,0,0,0,0,0,0,0,0,0,0,0,0,0,0,0,0) \\
   \tilde{t}= & (0,0,0,0,0,0,0,0,0,0,0,0,0,0,0,0,0,0,0,0,0,0,0,7,2,17,16,12,16,12,10,6,8,3,0,0,3,8,6,10,12,16,12,16,17,2)\end{align*}
   \normalsize
 It was verified that $sht=\tilde{s}h\tilde{t}$, so a legitimate private key was recovered.
\end{example}

\begin{example}

For parameters $n=19,\
q=19$, $\lambda=18$, and using the notations above, consider the randomly generated public element $h$, and secret key $(s,t)\in \fqa C_n\times \Gamma_\alpha$.
\scriptsize \begin{align*}
   h= & (14,5,13,4,10,12,8,6,17,18,15,1,14,14,15,15,13,4,6,7,7,11,13,4,11,12,3,11,18,8,3,3, 6,11,17,1,7,10), \\
   s=& (18,14,1,0,15,5,7,0,1,7,10,5,9,18,2,12,17,12,14,0,0,0,0,0,0,0,0,0,0,0,0,0,0,0,0,0,0,0) \\ 
   t = & (0,0,0,0,0,0,0,0,0,0,0,0,0,0,0,0,0,0,0,1,17,4,10,18,5,5,9,15,18,18,15,9,5,5,18,10,4,17)
\end{align*}\normalsize
Using the method in Section~\ref{rdxn}, the program computed the solution $(\tilde{s}, \tilde{t})$ to the DPD, where
\scriptsize
\begin{align*}
   \tilde{s} = & (12,6,4,10,12,4,5,7,0,15,8,7,1,0,2,15,6,7,1,0,0,0,0,0,0,0,0,0,0,0,0,0,0,0,0,0,0,0) \\
   \tilde{t}= & (0,0,0,0,0,0,0,0,0,0,0,0,0,0,0,0,0,0,0,14,13,11,10,3,3,1,1,3,16,16,3,1,1,3,3,10,11,13)\end{align*}
   \normalsize
 It was verified that $sht=\tilde{s}h\tilde{t}$, so a legitimate private key was recovered.

\end{example}

\begin{example}

For parameters $n=41,\
q=41$, $\lambda=29$, and using the notations above, consider the randomly generated public element $h$, and secret key $(s,t)\in \fqa C_n\times \Gamma_\alpha$.
\scriptsize \begin{align*}
   h= & (33,2,29,20,9,5,36,13,26,15,38,27,33,4,20,4,14,23,12,0,35,5,38,40,1,6,16,26,9,0,29,6,32,26,14,32,18,29, 13,35, 7, \\ & 8,38,26,20,25,24,18,30,28,22,8,21,1,33,29,2, 22,25,6,13,24,18,26,30,38,3,1,39,11,15,10,9,16,3,7,36,26,22,6,0,15), \\
   s=& (24,2,12,32,10,2,27,1,5,7,17,32,7,24,28,26,17,8,32,18,13,8,19,17,0,11,33,17,27,1,36,3,33,9,30,34,\\ & 22,26,21,5,29,0,0,0,0,0,0,0,0,0,0,0,0,0,0,0,0,0,0,0,0,0,0,0,0,0,0,0,0,0,0,0,0,0,0,0,0,0,0,0,0,0) \\ 
   t = & (0,0,0,0,0,0,0,0,0,0,0,0,0,0,0,0,0,0,0,0,0,0,0,0,0,0,0,0,0,0,0,0,0,0,0,0,0,0,0,0,0,13,8,18,11,31,\\ & 9,34,3,16,  39,32,0,15,31,3,26,0,31, 39,4,40,40,4,39,31,0,26,3,31,15,0,32,39,16,3,34,9,31,11,18,8)
\end{align*}\normalsize
Using the method in Section~\ref{rdxn}, the program computed the solution $(\tilde{s}, \tilde{t})$ to the DPD, where
\scriptsize
\begin{align*}
   \tilde{s} = & (39,9,4,23,8,8,10,40,31,27,22,36,11,14,35,28,25,0,0,10,16,33,24,6,33,17,15,13,17,10,18,31, 33,16,13,\\ & 28, 2,36,37,13,30,0,0,0,0,0,0,0,0,0,0,0,0,0,0,0,0,0,0,0,0,0,0,0,0,0,0,0,0,0,0,0,0,0,0,0,0,0,0,0,0,0) \\
   \tilde{t}= & (0,0,0,0,0,0,0,0,0,0,0,0,0,0,0,0,0,0,0,0,0,0,0,0,0,0,0,0,0,0,0,0,0,0,0,0,0,0,0,0,0,35,6,8,35,23,22,39,\\ & 12,22,36,34,1,29,8,16,40,29,16,24,14,31,31,14,24,16,29,40,16,8,29,1,34,36,22,12,39,22,23,35,8,6)\end{align*}
   \normalsize
 It was verified that $sht=\tilde{s}h\tilde{t}$, so a legitimate private key was recovered.
\end{example}

Clearly, in each of the above examples, $s\neq \tilde{s}$ and $t\neq \tilde{t}$, but $sht= \tilde{s}h\tilde{t}$. Thus, each of these examples also serves as a counterexample to the injectivity of the two-sided action. 

\section{Conclusion}
In this paper, we provided a method for cryptanalysis of the protocol in \cite{orig} which is based on a double-sided multiplication problem in the twisted dihedral group algebra $\fqa D_{2n}$. We first showed that the underlying DPD algorithmic problem is equivalent to a commutative semigroup action problem. For our cryptanalysis, we showed that the task for an adversary attempting to solve the underlying DPD problem is equivalent to the solution of two equations in $\fq$ involving circulant matrices. We further demonstrated a polynomial time solution for these equations using linear algebra, which works with a probability of $1-(1-\frac{1}{q})^2$. For the proposed values of the parameters in \cite{orig}, this gives a success rate of at least 90 percent. The key exchange system in \cite{orig} and its underlying algorithmic problem are both therefore clearly insecure, even in a classical setting.

\bibliographystyle{plain}
\bibliography{references}

\begin{thebibliography}{10}

\bibitem{ben2018cryptanalysis}
Adi Ben-Zvi, Arkadius Kalka, and Boaz Tsaban.
\newblock Cryptanalysis via algebraic spans.
\newblock In {\em Annual International Cryptology Conference}, pages 255--274.
  Springer, 2018.

\bibitem{canetti}
Ran Canetti and Hugo Krawczyk.
\newblock Analysis of key-exchange protocols and their use for building secure
  channels.
\newblock In Birgit Pfitzmann, editor, {\em Advances in Cryptology ---
  EUROCRYPT 2001}, pages 453--474, Berlin, Heidelberg, 2001. Springer Berlin
  Heidelberg.

\bibitem{cavallo}
Bren Cavallo and Delaram Kahrobaei.
\newblock A family of polycyclic groups over which the uniform conjugacy
  problem is np-complete.
\newblock {\em International Journal of Algebra and Computation},
  24(04):515--530, 2014.

\bibitem{orig}
Javier de~la Cruz and Ricardo Villanueva-Polanco.
\newblock Public key cryptography based on twisted dihedral group algebras.
\newblock {\em Advances in Mathematics of Communications}, 0:--, 2022.

\bibitem{twistedgpcodes}
Javier De~La~Cruz and Wolfgang Willems.
\newblock Twisted group codes.
\newblock {\em IEEE Transactions on Information Theory}, 67(8):5178--5184,
  2021.

\bibitem{efte}
Mohammad Eftekhari.
\newblock Cryptanalysis of some protocols using matrices over group rings.
\newblock pages 223--229, 04 2017.

\bibitem{Gu2014ConjugacySB}
Lize Gu and Shihui Zheng.
\newblock Conjugacy systems based on nonabelian factorization problems and
  their applications in cryptography.
\newblock {\em J. Appl. Math.}, 2014:630607:1--630607:10, 2014.

\bibitem{sym11081019}
María~Dolores Gómez~Olvera, Juan~Antonio López~Ramos, and Blas
  Torrecillas~Jover.
\newblock Public key protocols over twisted dihedral group rings.
\newblock {\em Symmetry}, 11(8), 2019.

\bibitem{mat-gpring}
Delaram Kahrobaei, Charalambos Koupparis, and Vladimir Shpilrain.
\newblock Public key exchange using matrices over group rings.
\newblock {\em Groups - Complexity - Cryptology}, 5(1):97--115, 2013.

\bibitem{gp-comm}
Juan~Antonio L{\'o}pez-Ramos, Joachim Rosenthal, Davide Schipani, and Reto
  Schnyder.
\newblock An application of group theory in confidential network
  communications.
\newblock {\em Mathematical Methods in the Applied Sciences}, 41:2294 -- 2298,
  2016.

\bibitem{maze}
Gérard Maze, Chris Monico, and Joachim Rosenthal.
\newblock Public key cryptography based on semigroup actions.
\newblock {\em Adv. in Math. of Communications}, 1(4):489--507, 2007.

\bibitem{menezes}
Alfred Menezes and Yihong Wu.
\newblock The discrete logarithm problem in {$GL(n, q)$}.
\newblock {\em Ars Comb.}, 47, 1997.

\bibitem{monico}
C.~Monico.
\newblock {\em \href{http://www.nd.edu/~rosen/preprints.html}{Semirings and
  Semigroup Actions in Public-Key Cryptography}}.
\newblock PhD thesis, University of Notre Dame, May 2002.

\bibitem{myasnikov2015linear}
Alexei Myasnikov and Vitaliĭ Roman'kov.
\newblock A linear decomposition attack.
\newblock {\em Groups Complexity Cryptology}, 7(1):81--94, 2015.

\bibitem{quantum-gpring}
Alexey~D. Myasnikov and Alexander Ushakov.
\newblock Quantum algorithm for discrete logarithm problem for matrices over
  finite group rings.
\newblock {\em Groups Complexity Cryptology}, 6(1):31--36, 2014.

\bibitem{rings}
Mar{\'i}a-Dolores Olvera-Lobo, Juan~Antonio L{\'o}pez-Ramos, and Blas
  Torrecillas.
\newblock Public key protocols over twisted dihedral group rings.
\newblock {\em Symmetry}, 11:1019, 2019.

\bibitem{Romankov2017AGE}
Vitaly Roman’kov.
\newblock A general encryption scheme using two-sided multiplications with its
  cryptanalysis.
\newblock {\em arXiv: Group Theory}, 2017.

\bibitem{Romankov-gen}
Vitaly Roman’kov.
\newblock Two general schemes of algebraic cryptography.
\newblock {\em Groups Complexity Cryptology}, 10(2):83--98, 2018.

\bibitem{prop1}
Simona Samardjiska, Paolo Santini, Edoardo Persichetti, and Gustavo Banegas.
\newblock A reaction attack against cryptosystems based on lrpc codes.
\newblock In Peter Schwabe and Nicolas Th{\'e}riault, editors, {\em Progress in
  Cryptology -- LATINCRYPT 2019}, pages 197--216, Cham, 2019. Springer
  International Publishing.

\bibitem{shor}
Peter~W. Shor.
\newblock Algorithms for quantum computation: discrete logarithms and
  factoring.
\newblock In {\em 35th Annual Symposium on Foundations of Computer Science
  (Santa Fe, NM, 1994)}, pages 124--134. IEEE Comput. Soc. Press, Los Alamitos,
  CA, 1994.

\bibitem{shpilrain2005new}
Vladimir Shpilrain and Alexander Ushakov.
\newblock A new key exchange protocol based on the decomposition problem.
\newblock {\em arXiv preprint arXiv:0512140}, 2005.

\bibitem{sagemath}
{The Sage Developers}.
\newblock {\em {S}ageMath, the {S}age {M}athematics {S}oftware {S}ystem
  ({V}ersion 8.6)}, 2020.
\newblock {\tt https://www.sagemath.org}.

\bibitem{complex}
Simran Tinani, Carlo Matteotti, and Joachim Rosenthal.
\newblock Complexity of conjugacy search in some polycyclic and matrix groups,
  2022.

\bibitem{tsaban2015polynomial}
Boaz Tsaban.
\newblock Polynomial-time solutions of computational problems in
  noncommutative-algebraic cryptography.
\newblock {\em Journal of Cryptology}, 28(3):601--622, 2015.

\end{thebibliography}

\end{document}